\documentclass[12pt]{article}
\usepackage{fullpage}

\usepackage{amsmath,amsthm,amssymb}
\usepackage{graphicx}
\usepackage{epstopdf}
\usepackage{algorithm}
\usepackage{algorithmic}
\usepackage{hyperref}

\usepackage{natbib}
\setcitestyle{authoryear,round,citesep={;},aysep={,},yysep={;}}

\usepackage{amsthm,amsmath,amssymb}
\usepackage{xcolor}

\newcommand{\ignore}[1]{}

\newtheorem{theorem}{Theorem}

\newtheorem{lemma}[theorem]{Lemma}
\newtheorem{corollary}[theorem]{Corollary}
\newtheorem{definition}[theorem]{Definition}

\newtheorem*{theorem*}{Theorem}
\newtheorem*{lemma*}{Lemma}
\newtheorem*{corollary*}{Corollary}
\newtheorem*{definition*}{Definition}
\newtheorem*{claim*}{Claim}
\newtheorem*{fact*}{Fact}

\DeclareMathOperator*{\conv}{conv}

\newcommand{\tr}{^{\top}}

\newcommand{\st}{\star}

\renewcommand{\t}[1]{\tilde{#1}}
\newcommand{\abs}[1]{\left|#1\right|}
\newcommand{\norm}[1]{\left\|#1\right\|}
\newcommand{\lr}[1]{\left(#1\right)}
\newcommand{\set}[1]{\left\{#1\right\}}

\newcommand{\reals}{\mathbb{R}}

\newcommand{\K}{\mathcal{K}}

\newcommand{\eps}{\epsilon}
\newcommand{\del}{\delta}
\newcommand{\Del}{\Delta}

\renewcommand{\O}{O}

\newcommand{\tO}{\t{\O}}

\newcommand{\trace}{\mathrm{tr}}

\newcommand{\xtil}{\tilde{x}}
\newcommand{\ytil}{\tilde{y}}

\newcommand{\evmax}{\lambda_{\max}}

\title{A Linear-Time  Algorithm for Trust Region Problems}

\author{%
Elad Hazan \\
\normalsize Technion \\
\normalsize\texttt{ehazan@ie.technion.ac.il}
\and
Tomer Koren \\
\normalsize Technion \\
\normalsize\texttt{tomerk@technion.ac.il}
}

\begin{document}
\maketitle

\begin{abstract}
We consider the fundamental problem of maximizing a general quadratic function over an ellipsoidal domain, also known as the trust region problem. We give the first provable linear-time (in the number of non-zero entries of the input) algorithm for approximately solving this problem. 
Specifically, our algorithm returns an $\eps$-approximate solution in time $\tilde{O}(N/\sqrt{\eps})$, where $N$ is the number of non-zero entries in the input.
This matches the runtime of Nesterov's accelerated gradient descent, suitable for the special case in which the quadratic function is concave, and the runtime of the Lanczos method which is applicable when the problem is purely quadratic.
\end{abstract}

\section{Introduction}

Perhaps the most elementary  quadratic optimization problem  is that of maximizing a quadratic function over the unit ball. This is precisely the {\em trust region} problem, or formally:
\begin{align} \label{eq:problem}
\begin{array}{ll}
	\text{maximize} \quad & x\tr A x + 2b \tr x \\
	\text{subject to} \quad & \norm{x}^{2} \le 1 ~,
\end{array}
\end{align}
where $A \in \reals^{n \times n}$ is an arbitrary $n \times n$ symmetric (possibly indefinite) matrix and $b \in \reals^n$. 

The trust region problem has numerous applications in optimization, where trust region methods \citep{conn2000trust} are among the most empirically successful techniques for solving nonlinear optimization problems.
In these methods, that enjoy strong convergence properties, at each iteration of the algorithm a quadratic approximation of the objective function is minimized over a ball, called the trust region.
Trust region problems are also useful in combinatorial optimization \citep{busygin2006new}, least-squares problems \citep{zhang2010derivative}, constrained eigenvalue problems \citep{gander1989constrained}, and more.

Despite being non-convex, the trust region problem has been shown to exhibit strong duality properties and is known to be solvable in polynomial time (e.g., \citealt{ben1996hidden,ye2003new}).
More specifically, it can be shown to be equivalent to more complex, albeit convex, semidefinite programming (SDP) optimization problems that can be solved using interior-point methods in polynomial time \citep{nesterov1994interior,alizadeh1995interior}.
However, the worst-case complexity of current interior-point based solvers for SDP problems is a large polynomial, and they are inefficient in exploiting sparsity of the data.
As a consequence, this approach not practical for large-scale problems.

On the other hand, the close connections between the trust region problem and eigenvalue problems 
suggest that more efficient trust region algorithms should exist.
Indeed, if the problem is purely quadratic, i.e., when $b=0$, then the trust region problem reduces to the fundamental maximal eigenvector problem, that can be approximated in linear time via the well-known Power and Lanczos methods.
This observation has led authors in the optimization and the numerical analysis communities to develop efficient, matrix-free algorithms that are based solely on matrix-vector products.
Notable examples are the dual-based algorithms of \cite{more1983computing}, \cite{rendl1997semidefinite} and \cite{sorensen1997minimization}, the generalized Lanczos trust-region method of \cite{gould1999solving}, and the more modern advancements of \cite{rojas2001new,erway2009subspace,erway2009iterative,gould2010solving}.
However, while being provably convergent in most cases, the runtime evaluation of these algorithms is essentially empirical and lacks formal guarantees.
To the best of our knowledge, to date there is no formal evidence that the trust region problem can be solved in linear time (in the worst case), as the closely-related maximal eigenvalue problem.

The main hurdle faced by most previous methods is a certain case in which numerical difficulties arise, so-called the ``hard case'' \citep{more1983computing}, and most research in the last two decades on the trust region problem focuses on addressing this issue. 
This phenomenon occurs when the linear component vector $b$ is nearly orthogonal to the eigenspace of the smallest eigenvalue of $A$,
and seems to be the reason for the lack of provable worst-case convergence bounds for the trust region problem.

\paragraph{Our contribution.}
In this paper we show that the additional linear term and the non-convexity of the problem do not add to its complexity:
we devise a novel linear-time algorithm for approximating the trust region problem that has, up to logarithmic terms, the same worst-case time complexity of \emph{a single maximal eigenvalue computation}.  
Our approach reduces the trust region problem into a series of eigenvalue computations, thus it is able to exploit data sparsity and runs in time linear in the number of non-zero entries of the input.
In the specific case where the problem is convex (i.e., when $A$ is negative semidefinite), the same complexity guarantees can be obtained by applying Nesterov's accelerated gradient descent \citep{nesterov1983method} to the problem; thus, our approach can be seen as an analog of the latter algorithm to the general non-convex case.

 Our approach is based on an SDP relaxation of a feasibility version of the trust region problem.
While SDP relaxations are already standard for this problem (e.g., \citealt{rendl1997semidefinite, ye2003new}), our approach uses a specific form of SDP that can be approximated quickly via eigenvalue computations and does not require us to use interior-point solvers.
Another important feature of our relaxation is that it admits an efficient and accurate rounding procedure; that is, given a matrix solution to the SDP we are able to extract an approximate vector solution to the original trust region problem of almost the same quality.
This rounding procedure allows us to avoid numerical problems and complications that exist in previous approaches, which occur for certain configurations of the eigenvalues of $A$, namely the ``hard case'' mentioned above.

At the heart of our approach is an efficient, linear-time solver for the SDP relaxation.
This solver exploits the special structure of the dual problem, which is essentially a one-dimensional problem for which bisection techniques can be applied.
Each dual step of this algorithm, that converges in a logarithmic number of iterations, amounts to a single approximate eigenvalue computation.
The major technical difficulty that results from working in the dual, is in obtaining a primal solution from the dual iterations.
To this end, we employ a technique reminiscent of the ``ellipsoid against hope'' algorithm \citep{papadimitriou2008computing}, used by \cite{arora2005fast} in the context of approximate semi-definite programming, for recovering a primal solution by solving a small linear program formed by the dual iterates.

\section{Setup and Statement of Results}

In this section we formalize the setting and state our main results.
We shall describe our algorithm and results in a slightly more general setting, that include optimization problems of the form:
\begin{align} \label{eq:M-problem}
\begin{array}{ll}
	\text{maximize} \quad & x\tr A x + 2b \tr x \\
	\text{subject to} \quad & \norm{x}_{M}^{2} \le 1 ~,
\end{array}
\end{align}
in which the optimization domain is an ellipsoidal set, described by a general norm constraint $\norm{x}_{M} \le 1$.
Here, $\norm{x}_{M} = \sqrt{x\tr M x}$ is a norm induced by a positive definite matrix $M \in \reals^{n \times n}$.
We note that the optimal solution $v^{\st}$ to this problem is necessarily non-negative, as the objective function equals zero for $x = 0$ and the value at the optimum can only be larger.

For our bounds, we use the following notation.
We let
\begin{equation*}
\begin{aligned}
	\lambda &= \max \set{ 2(\norm{A}_{2} + \norm{b}), \norm{M}_{2}, 1 } ~, \\
	\mu &= \min \set{ \lambda_{\min}(M), 1 } ~,
\end{aligned}
\end{equation*}
where $\norm{\cdot}$ is the Euclidean vector norm, the matrix norm $\norm{\cdot}_{2}$ is the spectral norm, and $\lambda_{\min}(\cdot)$ and $\evmax(\cdot)$ refer to the minimal and maximal eigenvalues of a matrix.
We refer to the ratio $\kappa = \lambda/\mu$ as the ``condition number'' of the problem%
\footnote{Notice that this condition number can be approximated in linear time (see Section~\ref{sub:oracle}).}.
The objective is $\lambda$-Lipschitz, so the optimal value $v^\st$ lies in the interval $[0,\lambda]$.
For our runtime results, we also let~$N$ be an upper bound over the number of non-zero entries in the matrices $A$ and $M$, and assume without loss of generality that~$N \ge n$.

As mentioned earlier, our goal is to reduce the approximation of problem~\eqref{eq:M-problem} in the general case to a series of approximate eigenvalue computations, and thereby
obtain an algorithm that runs in time linear in the number of non-zero entries $N$ in the matrices~$A$ and $M$.
To this end, we formally define the notion of an approximate eigenvalue oracle.

\begin{definition} \label{def:oracle}
An approximate eigenvalue oracle is a randomized procedure that given a matrix $A \in \reals^{n \times n}$ and parameters $\eps,\del > 0$, with probability at least $1-\del$ returns a vector $x \in \reals^n$ such that $x\tr A x \ge \evmax(A) - \eps$.
\end{definition}

An approximate eigenvalue oracle can be implemented to run in linear time via the Lanczos method. 
For completeness, this is formally shown in Section~\ref{sub:oracle}.

Our algorithm solves the equivalent feasibility problem
\begin{align} \label{eq:intro-feas}
\begin{array}{lll}
	x\tr A x + 2b \tr x &\ge& c + \eps/\kappa \\
	\norm{x}_{M}^{2} &\le& 1 - \eps/\kappa ~,
\end{array}
\end{align}
with $c \in [0,\lambda]$.
Indeed, we can reduce the task of $\eps$-approximating the trust region problem to~\eqref{eq:intro-feas} via a binary search over the optimum~$v^{\st}$ (incurring an additional log-factor to the running time), with $c$ being our current guess of $v^\st$.
Since $\eps/\kappa \le \eps$, tightening the first constraint by $\eps/\kappa$ only harms the approximation by a constant factor.
Note that we have also relaxed the second constraint by $\eps/\kappa$, but this only shifts the optimal solution by an Euclidean distance of at most $(1/\mu)\eps/\kappa$, which in turn translates to an additional $(\lambda/\mu) \eps/\kappa = \eps$ bias in the objective value (since the objective function is $\lambda$-Lipschitz).

More specifically, the algorithm we describe approximately solves \eqref{eq:intro-feas}, in the sense that it either correctly declares that the problem is infeasible or finds a vector $x \in \reals^{n}$ such that
\begin{align} \label{eq:intro-feas-approx}
\begin{array}{lll}
	x\tr A x + 2b \tr x &\ge& c \\
	\norm{x}_{M}^{2} &\le& 1 ~.
\end{array}
\end{align}
The main result of this manuscript is the following.

\begin{theorem} \label{thm:main}
Given an approximate eigenvalue oracle and parameters $\eps, \del > 0$, with probability at least $1-\del$ (over the randomization of the oracle) Algorithm~\ref{alg:main} below returns a vector $x \in \reals^n$ for which~\eqref{eq:intro-feas-approx} holds or correctly declares that~\eqref{eq:intro-feas} is infeasible.
The algorithm invokes the oracle  $\O(\log(\kappa/\eps))$ times and can be implemented to run in total $\tO(N\sqrt{\kappa/\eps})$ time%
\footnote{Throughout, we use the $\tilde{O}$ notation to hide constant and poly-logarithmic factors.}.
\end{theorem}

A proof of the theorem is provided in Section~\ref{sub:analysis}.

\section{The Algorithm} \label{sec:algorithm}

In this section we describe our linear-time algorithm for the feasibility problem~\eqref{eq:problem},
which is summarized in Algorithm~\ref{alg:main}.

\begin{algorithm} \caption{$\texttt{TrustRegion}\,(A,b,M,\eps)$ \label{alg:main}}
\begin{algorithmic}[1]

\INPUT $A,M \in \reals^{n \times n}$, $b \in \reals^{n}$ and $\eps>0$
\OUTPUT a vector $x \in \reals^{n}$ satisfying \eqref{eq:intro-feas-approx}, or infeasibility of \eqref{eq:intro-feas}\\[0.2cm]

\STATE let $\eps' = \eps/2\kappa$
\STATE define the $(n+1) \times (n+1)$ matrices
\begin{align} \label{eq:A1A2}
	A_1 = \frac{1}{2\kappa} \cdot \lr{\begin{array}{cc} -c & b\tr \\b & A\end{array}},
	\qquad
	A_2 = \frac{1}{2\kappa} \cdot \lr{\begin{array}{cc} 1 & 0 \\0 & -M\end{array}}
\end{align}

\STATE invoke $\texttt{SolveSDP}(A_{1},A_{2},\eps'/2)$ that returns $X = \sum_{i=1}^{r} x_{i} x_{i}\tr$ as output

\IF {\texttt{SolveSDP} returned ``\texttt{infeasible}''}
\STATE conclude that \eqref{eq:intro-feas} is infeasible
\ELSE
\STATE invoke $\texttt{SZRotation}(A_{1},X)$ that returns $X = \sum_{i=1}^{r} y_{i} y_{i}\tr$ as output
\STATE find a vector $y \in \set{y_{1},y_{2},\ldots,y_{r}}$ for which $y\tr A_{2} y \ge \eps'/4r$
\STATE let $\alpha = y(1)$ be the first entry of $y$ and let $\t{y} = (y(2),y(3),\ldots,y(n+1))$
\STATE \textbf{return} $x = \t{y}/\alpha$
\ENDIF 

\end{algorithmic}
\end{algorithm}

\paragraph{SDP relaxation.}
The first step in our approach is to relax the feasibility problem \eqref{eq:intro-feas} to the following SDP program:
\begin{align} \label{eq:intro-sdp}
\begin{array}{l}
	A_i \bullet X \ge \eps'/2 \qquad (i=1,2), \\
	X \in \K_{n+1} ~,
\end{array}
\end{align}
with $A_{1},A_{2}$ given in Eq.~\eqref{eq:A1A2}, $\eps' = \eps/\kappa$ and
\begin{align*}
	\K_{n} = \set{X \in \reals^{n \times n} ~:~ X \succeq 0,\, \trace(X) \le 1} 
\end{align*}
being the set of all $n \times n$ positive semidefinite matrices with trace at most one.
Here and henceforth, we use the notation $A \bullet X = \text{trace}(A\tr X)$ to denote dot-products of matrices.
To see that this is indeed a relaxation, note that given any feasible solution $x$ to \eqref{eq:intro-feas}, the rank-one matrix $X = \tfrac{1}{2} (\,1,~ x\,) \cdot (\,1,~ x\,)\tr$, with trace $\trace(X) = \frac{1}{2}(1+\norm{x}_2^2) \le 1$, is feasible for the SDP.

This specific form of SDP has two important advantages.
First, the structure of $\K_{n+1}$ allows us to use the eigenvalue oracle for solving the SDP, as a linear optimization over this set is equivalent to a maximal eigenvalue computation.
Second, as we show in our analysis, the artificial extra slack of $\eps'/2$ we imposed in the constraints ensures that a solution to this relaxation has sufficient mass on its first entry, which makes it possible to avoid the ``hard case'' entirely.

\paragraph{Linear-time SDP solver.}

The main ingredient in our approach, which is described in Section~\ref{sec:sdp-solver}, is a linear-time procedure \texttt{SolveSDP} for solving SDP programs of the form~\eqref{eq:intro-sdp}.
As discussed above, generic SDP solvers are not suitable for this task as they are not able to exploit the sparsity of the input and run in super-linear time.
In order to approximate this problem quickly, we avoid solving the primal problem directly and instead attack the dual problem to \eqref{eq:intro-sdp}:
\begin{align*}
	& p \cdot A_1 \bullet X + (1-p) \cdot A_2 \bullet X  < \eps'/2 \qquad \forall ~ X \in \K_{n+1} ~, \\
	& 0 \le p \le 1 ~.
\end{align*}
We show that this one-dimensional problem can be solved quickly via a binary search that in each iteration invokes an approximate eigenvalue oracle on a matrix of the form $p A_1 + (1-p) A_2$.
Moreover, whenever the primal SDP \eqref{eq:intro-sdp} is feasible, we show how to efficiently recover an approximate solution, namely a matrix $X \in \K_{n+1}$ such that $A_i \bullet X \ge \eps'/4$, from the outputs of the oracle calls along the execution of the binary search.
Formally, in Section~\ref{sec:sdp-solver} we prove:

\begin{lemma} \label{lem:sdp}
Given access to an approximate eigenvalue oracle, with probability at least $1-\del$ (over the randomization of the oracle) \texttt{SolveSDP} outputs a decomposition $X = \sum_{i=1}^{r} x_i x_i\tr$ of a matrix $X \in \K_{n+1}$ of rank $r=2$ for which $A_i \bullet X \ge \eps'/4 \;\; (i=1,2)$, or correctly declares that \eqref{eq:intro-sdp} is infeasible.
The algorithm calls the oracle at most~$\O(\log(\kappa/\eps))$ times and can be implemented to run in total time~$\tO(N\sqrt{\kappa/\eps})$.
\end{lemma}

An important feature of our SDP solver is that it produces a solution of a very low rank, namely with $r=2$ (it can be shown that a rank-2 solution for formulation \eqref{eq:intro-sdp} does exist, via duality conditions). 
As we explain below, this would allow us to recover a rank-$1$ solution quickly, and in turn obtain a feasible solution to \eqref{eq:intro-feas} in linear time. 
Nevertheless, we preserve full generality and present our algorithm with arbitrary rank $r$, as our approach may work with any other SDP solver that might produce solutions of rank higher than two.

\paragraph{Rounding of SDP.}

Finally, we show how to obtain an approximate solution to \eqref{eq:intro-feas} from our solution $X$ to the relaxation, with linear-time computations.
The method we describe is based on a matrix rotation procedure, given in Algorithm~\ref{alg:decomp}, which is a variant of a procedure due to \cite{sturm2003cones}.

\begin{algorithm} \caption{\texttt{SZRotation}\,$(A,X)$ \label{alg:decomp}}
\begin{algorithmic}[1]

\INPUT {matrices $A$ and $X = \sum_{i=1}^{r} x_{i} x_{i}\tr$ such that $A \bullet X \ge a$}\\
\OUTPUT {matrix $X = \sum_{i=1}^{r} x_{i} x_{i}\tr$ such that $A \bullet x_{i}x_{i}\tr \ge a/r$ for all $i$}\\[0.2cm]

\STATE let $a' = A \bullet X$

\WHILE {there exist $x_i,x_j$ such that $x_i\tr A x_i > a'/r$, $x_j\tr A x_j < a'/r$}
	\STATE compute a root $t$ of the quadratic equation
		\begin{align} \label{eq:quad}
			\lr{x_i\tr A x_i - \frac{a'}{r}} \cdot t^2
			+ (2 x_i\tr A x_j) \cdot t
			+ \lr{x_j\tr A x_j - \frac{a'}{r}} = 0
		\end{align}
	\STATE replace the vectors $x_{i},x_{j}$ with the vectors 
	$
		\t{x}_{i} = \frac{1}{\sqrt{t^2 + 1}} (t x_i + x_j),~
		\t{x}_{j} = \frac{1}{\sqrt{t^2 + 1}} (x_i - t x_j)
	$
\ENDWHILE

\STATE {\bf return} $X = \sum_{i=1}^{r} x_{i} x_{i}\tr$

\end{algorithmic}
\end{algorithm}

The procedure provides the following guarantee, which is proved in Section~\ref{sub:analysis} below.

\begin{lemma} \label{lem:decomp}
Given a decomposition $X = \sum_{i=1}^r x_i x_i\tr$ of a positive semidefinite matrix $X$ of rank $r$ and an arbitrary matrix $A$ with  $A \bullet X \ge a$, \texttt{SZRotation} outputs a decomposition $X = \sum_{i=1}^r y_i y_i\tr$ such that $y_i \tr A y_i \ge a/r$ for all $i \in [r]$.
The procedure runs in time $\O(N r)$, where $N \ge n$ is the number of non-zero entries in $A$.
\end{lemma}

In particular, for being able to recover a solution in linear time, it is essential to begin with a solution to the SDP of a constant rank---which is provided by our SDP solver.
Employing this decomposition, one can compute an approximate solution to the feasibility problem~\eqref{eq:intro-feas}, namely a vector $y$ satisfying~\eqref{eq:intro-feas-approx}.

\subsection{Analysis} \label{sub:analysis}

In the rest of this section we  prove  our main theorem (Theorem~\ref{thm:main}).
First, we prove Lemma~\ref{lem:sdp} using the results of Section~\ref{sec:sdp-solver}.

\begin{proof}[Proof of Lemma~\ref{lem:sdp}]
In order to show that the lemma follows from Theorem~\ref{thm:sdp} below, we have to show that for any $c \in [0, \lambda]$, the spectral norm of the matrices $A_{1}$ and $A_{2}$ defined in Eq.~\eqref{eq:A1A2} is at most $1$.

We first consider the matrix $A_{1}$.
Let $x \in \reals^{n+1}$ be some unit vector and write $x = (\alpha, y)$ with $\alpha \in \reals$ and $y \in \reals^{n}$.
Since $x$ has unit norm, $\abs{\alpha} \le 1$ and $\norm{y}^{2} \le 1$, thus
\begin{align*}
	2\kappa \cdot | x\tr A_{1} x |
	&= | -c\alpha^{2} + 2\alpha b\tr y + y\tr A y \, | \\
	&\le c\alpha^{2} + 2 |\alpha| \norm{b}\norm{y} + \norm{A}_{2} \norm{y}^{2} \\
	&\le c + 2 (\norm{b} + \norm{A}_{2}) \\
	&\le 2\lambda ~,
\end{align*}
where the last inequality follows from the fact that $2(\norm{A}_{2} + \norm{b}) \le \lambda$.
Since $\lambda \le \kappa$ and the above applies for any unit vector $x$, we have shown that $\norm{A_{1}}_{2} \le 1$.

Similarly, for the matrix $A_{2}$ we have
\begin{align*}
	2\kappa \cdot | x\tr A_{2} x |
	= | -\alpha^{2} + y\tr M y |
	\le 1 + \norm{M}_{2}
	\le 2\lambda
\end{align*}
for any unit vector $x$, which implies that $\norm{A_{2}}_{2} \le 1$.
\end{proof}

Next, we prove Lemma~\ref{lem:decomp}.

\begin{proof}[Proof of Lemma~\ref{lem:decomp}]
First, note that equation \eqref{eq:quad} has real roots since $x_i\tr A x_i - a'/r > 0$ and $x_j\tr A x_j - a'/r < 0$.
One can  verify that $\xtil_i \xtil_i\tr + \xtil_j \xtil_j\tr = x_i x_i\tr + x_j x_j\tr$, so that the equality $X = \sum_{i=1}^r x_i x_i\tr$ remains true along the execution of the algorithm.
On the other hand, note that
\begin{align*}
	\xtil_i\tr A \xtil_i 
	&= \frac{(t x_1 + x_2)\tr A (t x_1 + x_2)}{t^2+1} \\
	&= \frac{x_1\tr A x_1 \cdot t^2
		+ 2 x_1\tr A x_2 \cdot t
		+ x_2\tr A x_2}{t^2+1} \\
	&= \frac{a'}{r} ~,
\end{align*}
where the final equality follows from $t$ being a root of \eqref{eq:quad}.
Hence, each iteration produces an additional index $i$ for which $x_i\tr A x_i = a'/r$, and after at most $r$ iterations it must be the case that $x_i\tr A x_i = a'/r \ge a/r$ for all $i$.
Consequently, the total runtime is $\O(N r)$ as each iteration needs~$\O(N)$ time.
\end{proof}

We can now prove our main theorem.

\begin{proof}[Proof of Theorem~\ref{thm:main}]

We first prove the runtime guarantee, and then show the correctness of the algorithm.

\paragraph{Running time.}
Note that the number of non-zero entries in each of the matrices $A_{1}, A_{2}$ is $O(N)$.
Hence, according to Lemma~\ref{lem:sdp}, the call to \texttt{SolveSDP} in line~2 of the algorithm invokes the approximate eigenvalue oracle at most $\O(\log(\kappa/\eps))$ times and returns in time $\tO(N\sqrt{\kappa/\eps})$.
Whenever the SDP relaxation is feasible \texttt{SolveSDP} produces a solution of rank $r=2$, in which case Lemma~\ref{lem:decomp} states that the invocation of \texttt{SZRotation} in line~6 runs in time $\O(N)$.
Therefore, the total runtime of the entire algorithm is $\tO(N\sqrt{\kappa/\eps})$.

\paragraph{Correctness.}
First, assume that the algorithm returns ``\texttt{infeasible}''.
By the guarantees of \texttt{SolveSDP} (Lemma~\ref{lem:sdp}), this means that the SDP relaxation \eqref{eq:intro-sdp} is infeasible.
Hence, the problem~\eqref{eq:intro-feas} is also infeasible (otherwise, as explained above, we could convert a feasible solution $x$ of \eqref{eq:intro-feas} to a feasible solution $X$ to the SDP relaxation), as required in this case.

Next, assume that the algorithm returns a solution vector $x \in \reals^{n}$.
In this case, \texttt{SolveSDP} returns a matrix $X$ such that $A_i \bullet X \ge \eps'/4$ ($i=1,2$).
Hence, by Lemma~\ref{lem:decomp}, the output $X = \sum_{i=1}^{r} y_{i} y_{i}\tr$ of \texttt{SZRotation} has $y_{i} \tr A_{1} y_{i} \ge \eps'/4r$ for all $i=1,2,\ldots,r$.
On the other hand, since $A_2 \bullet X \ge \eps'/4$, at least one of the vectors $y_{i}$ must satisfy $y_i \tr A_2 y_i \ge \eps'/4r$.
Hence, the algorithm can indeed find a vector $y \in \set{x_1,\ldots,x_r}$ such that $y \tr A_2 y \ge \eps'/4r$, for which we also have $y\tr A_1 y \ge \eps'/4r$.
Rewriting the last two inequalities in terms of $\alpha$ and~$\ytil$, we get
\begin{align*}
	\ytil\tr A \ytil + 2\alpha b\tr \ytil - c \alpha^2 \ge \eps'/4r 
	\quad\text{and}\quad
	\alpha^2 - \norm{\ytil}_{M}^2 \ge \eps'/4r .
\end{align*}
The second inequality implies that $\alpha \ne 0$, and so we can divide through by $\alpha^2$. 
Note that since $X \in \K_{n+1}$ we also have $\alpha^2 \le \alpha^{2} + \norm{\t{y}}^{2} = \norm{y}^2 \le \trace(X) \le 1$, which implies that
the vector $x = \ytil/\alpha$ has
\begin{align*}
	x\tr A x + 2 b\tr x - c \ge \eps'/4r
	\quad\mbox{and}\quad
	1 - \norm{x}_{M}^2 \ge \eps'/4r \,.
\end{align*}
This in particular means that
\begin{align*}
	x\tr A x + 2 b\tr x \ge c
	\quad\mbox{and}\quad
	\norm{x}_{M}^2 \le 1 \,,
\end{align*}
as required.
\end{proof}

\section{Solving the Relaxation in Linear Time} \label{sec:sdp-solver}

In this section we describe a linear-time algorithm for approximately solving an SDP problem of the form
\begin{align} \label{eq:sdp-feas}
\begin{array}{lll}
	A_i \bullet X \ge \eps \qquad (i=1,2) \\
	X \in \K_{n} ~,
\end{array}
\end{align}
where $A_i \in \reals^{n \times n}$, $\norm{A_i}_2 \le 1$ ($i=1,2$) and the number of non-zero entries in each of the matrices $A_1,A_2$ is at most $N \ge n$.
Namely, the algorithm we present either finds a matrix $X \in \K_{n}$ for which $A_i \bullet X \ge \eps/2$ ($i=1,2$), or correctly declares that \eqref{eq:sdp-feas} is infeasible.

Our algorithm, given in Algorithm~\ref{alg:sdp}, applies as a first step a binary search for approximately solving the dual feasibility problem:
\begin{align*}
	& A(p) \bullet X < \eps/2 \qquad \forall X \in \K_{n}, \\
	& 0 \le p \le 1 ~,
\end{align*}
where we denote $A(p) = p A_1 + (1-p) A_2$ for all $p \in [0,1]$.
The binary search either correctly declares that the dual problem is infeasible or finds $p \in [0,1]$ such that $A(p) \bullet X < \eps$ for all $X \in \K_{n}$.
Infeasibility of the dual implies feasibility of the primal problem \eqref{eq:sdp-feas}, in which case our algorithm is able to efficiently recover a primal solution from the binary search iterates by solving a simple linear program.

The algorithm assumes the availability of an approximate eigenvector computation oracle, denoted \texttt{ApproxEV}, as described in Definition~\ref{def:oracle}.
In Section~\ref{sub:oracle} we explain how such an oracle can be implemented in time $\tO(N/\sqrt{\eps})$, where $N$ is the number of non-zero entries in the input matrix and $\eps$ is the error parameter.

\begin{algorithm} \caption{$\texttt{SolveSDP}\,(A_{1},A_{2},\eps)$ \label{alg:sdp}}
\begin{algorithmic}[1]

\INPUT matrices $A_1, A_2 \in \reals^{n}$ with $\norm{A_i}_2 \le 1$ and $\eps>0$
\OUTPUT matrix $X \in \K_{n}$ such that $A_i \bullet X \ge \eps/2$ ($i=1,2$), or infeasibility of \eqref{eq:sdp-feas} \\[0.2cm]

\STATE initialize~ $T \gets \log_2(8/\eps)$, $p_1 \gets 1, p_2 \gets 1$
\FOR {$t=1$ to $T$}
	\STATE let~ $p \gets (p_1 + p_2)/2$
	\STATE invoke~ $(\lambda, x) \gets \texttt{ApproxEV}(p A_1 + (1-p) A_2,~ \eps/4,~ \del/T)$
	\IF {$\lambda < 3\eps/4$}
		\STATE \textbf{return} ``\texttt{infeasible}''
	\ELSIF {$x\tr A_1 x < x\tr A_2 x$}
		\STATE update~ $p_1 \gets p, \; x_1 \gets x$
	\ELSE[if $x\tr A_1 x > x\tr A_2 x$]
		\STATE update~ $p_2 \gets p, \; x_2 \gets x$
	\ENDIF
\ENDFOR

\STATE compute $0 \le q \le 1$ such that
\begin{align} \label{eq:q}
	q \cdot x_1\tr A_i x_1 + (1-q) \cdot x_2\tr A_i x_2 \ge \eps/2 \qquad (i=1,2)
\end{align}

\STATE \textbf{return} $X = q \cdot x_1 x_1\tr + (1-q) \cdot x_2 x_2\tr$

\end{algorithmic}
\end{algorithm}

We now state the main result of this section.

\begin{theorem} \label{thm:sdp}
Given matrices $A_1, A_2 \in \reals^{n \times n}$ with $\norm{A_i}_2 \le 1$ and parameters $\eps, \del>0$, with probability at least $1-\del$ \texttt{SolveSDP} outputs a matrix $X \in \K_{n}$ of rank $2$ that satisfies $A_i \bullet X \ge \eps/2 \;\; (i=1,2)$, or correctly declares that \eqref{eq:sdp-feas} is infeasible.
The algorithm calls the orcale \texttt{ApproxEV} at most~$\O(\log(1/\eps))$ times and can be implemented to run in total time~$\tO(N/\sqrt{\eps})$.
\end{theorem}


For proving Theorem~\ref{thm:sdp} we need two simple duality results. 

\begin{lemma} \label{lem:duality}
Let $\mathcal{S}$ be an arbitrary compact set of matrices. 
Exactly one of the following statements holds:
\begin{enumerate}
\item
There exists $X \in \conv \mathcal{S}$ such that $A_i \bullet X \ge \eps \;\; (i=1,2)$.
\item
There exist $p \in [0,1]$ such that $A(p) \bullet X = (p A_1 + (1-p) A_2) \bullet X < \eps$ for all $X \in \mathcal{S}$.
\end{enumerate}
\end{lemma}

\begin{proof}
First assume that the first statement holds true, i.e.~there exists $X^\st \in \conv \mathcal{S}$ for which $A_i \bullet X \ge \eps \;\; (i=1,2)$.
Then, for all $p \in [0,1]$ we have
\[
	A(p) \bullet X^\st
	= p (A_1 \bullet X^\st) + (1-p) (A_2 \bullet X^\st)
	\ge \eps
\]
which proves that the second statement cannot hold.
Conversely, if there does not exist such~$X^\st$, then $\min_{p \in [0,1]} A(p) \bullet X = \min_{i} A_i \bullet X < \eps$ for all $X \in \conv \mathcal{S}$.
Applying Sion's minimax theorem \citep{sion1958general}, we obtain
\[
	\min_{p \in [0,1]} \max_{X \in \conv\mathcal{S}} A(p) \bullet X
	= \max_{X \in \conv\mathcal{S}} \min_{p \in [0,1]} A(p) \bullet X
	< \eps \,.
\]
This means that there exists $p^\st \in [0,1]$ such that $A(p^\st) \bullet X < \eps$ for all $X \in \mathcal{S}$, that is, the second statement is true.
\end{proof}

When $\mathcal{S} = \set{X_1,X_2,\ldots,X_m}$ is a finite set, we obtain the following corollary that enables us to compute an approximate primal solution from the binary search iterates.

\begin{corollary} \label{cor:farkas}
Let $X_1,X_2,\ldots,X_m$ be arbitrary matrices and assume that there does not exist $p \in [0,1]$ such that $A(p) \bullet X_i < \eps/2$ for all $i=1,2,\ldots,m$.
Then there exists $q \in \Del_m$ such that for $X = \sum_{i=1}^{m} q_i X_i$ it holds that $A_i \bullet X \ge \eps/2 \;\; (i=1,2)$.
\end{corollary}

We can now provide a proof of Theorem~\ref{thm:sdp}.

\begin{proof}[Proof of Theorem~\ref{thm:sdp}]~
\paragraph{Running time.} 
The algorithm invokes \texttt{ApproxEV} at most $T = \O(\log(1/\eps))$ times on matrices of the form $A(p)$.
Since by the triangle inequality $\norm{A(p)}_2 \le p \norm{A_1}_2 + (1-p) \norm{A_2}_{2} \le 1$ (as we assume that $\norm{A_1}_2, \norm{A_2}_2 \le 1$), Lemma~\ref{lem:approx-ev} shows that this procedure can be implemented using the Lanczos method to run in $\tO(N/\sqrt{\eps})$.
Hence, the entire algorithm runs in $\tO(N/\sqrt{\eps})$ time as all other operations are linear in the problem dimensions.

\paragraph{Correctness.}
First note that since each call to the oracle \texttt{ApproxEV} errs with probability at most $\del/T$, with probability at least $1-\del$ all oracle invocations return a correct output.
Observe that if the algorithm declares that the problem is infeasible, then there exists a value of $p$ for which \texttt{ApproxEV} outputs $\lambda < 3\eps/4$.
This means that
\begin{align*}
	\max_{X \in \K_{n}} A(p) \bullet X
	= \evmax(A(p))
	\le \lambda + \eps/4
	< \eps \,,
\end{align*}
that is, $A(p) \bullet X < \eps$ for all $X \in \K_{n}$.
Lemma~\ref{lem:duality} then implies that \eqref{eq:sdp-feas} is indeed infeasible.
Thus, we henceforth assume that all invocations of \texttt{ApproxEV} returned $\lambda \ge 3\eps/4$.

Consider the values of $p_1,p_2$ and $x_1,x_2$ at the end of the main loop of the algorithm, and denote $X_i = x_i x_i\tr$ ($i=1,2$). 
Then according to our assumption, we have $A(p_1) \bullet X_1 \ge 3\eps/4$ and $A(p_2) \bullet X_2 \ge 3\eps/4$.
Also, we have have $p_2 - p_1 \le \eps/8$ since the binary search continues for $\log_2(8/\eps)$ iterations.
Our central observation is that the system of two inequalities $A(p) \bullet X_i < \eps/2$ ($i=1,2$) must be infeasible (in $p$).
Indeed, assume that there exists some $p^\st \in [0,1]$ for which $A(p^\st) \bullet X_1 < \eps/2$ and $A(p^\st) \bullet X_2 < \eps/2$.
Notice that $A(p) \bullet X_1 \ge 3\eps/4$ for $p < p_1$, as $A(p_1) \bullet X_1 \ge 3\eps/4$ and the function $p \mapsto A(p) \bullet X_1$ is monotonically deceasing (recall that $x_1\tr A_1 x_1 < x_1\tr A_2 x_1$).
Similarly we have $A(p) \bullet X_2 \ge 3\eps/4$ for $p > p_2$, so it must be the case that $p^\st \in [p_1,p_2]$.
But this implies that 
\begin{align*}
	A(p_1) \bullet X_1 - A(p^\st) \bullet X_1
	= (p^\st - p_1) \cdot (A_2 \bullet X_1 - A_1 \bullet X_1)
	\le (p_2 - p_1) \cdot 2
	\le \eps/4
\end{align*}
from which we get that $A(p^\st) \bullet X_1 \ge A(p_1) \bullet X_1 - \eps/4 \ge \eps/2$ which is a contradiction to the choice of $p^\st$.

This infeasibility, together with Corollary~\ref{cor:farkas}, leads to the conclusion that there exists $q \in [0,1]$ such that $X = q X_1 + (1-q) X_2 \in \K_{n}$ satisfies $A_i \bullet X \ge \eps/2 \;\; (i=1,2)$.
Since the existence of such $q$ is guaranteed, it can be found by solving the simple linear program~\eqref{eq:q}, thereby retrieving a rank-two matrix which is an approximate solution to \eqref{eq:sdp-feas}.
\end{proof}

\subsection{Approximate Eigenvalue Computation} \label{sub:oracle}



In our analysis we require a linear-time procedure for approximating eigenvalues of sparse matrices.
The following lemma states that the \emph{Lanczos method} provides that.

\begin{lemma} \label{lem:approx-ev}
There exists an algorithm that given a matrix $M \in \reals^{n \times n}$ with $\norm{M}_2 \le \lambda$ and parameters~$\eps,\del > 0$, runs in time 
$$
	O\left({\frac{N \sqrt{\lambda}}{\sqrt{\epsilon}}\ln\frac{n}{\delta}}\right)
$$
where $N$ is the number of non-zero entries in the matrix $M$, and returns a unit vector $x \in \reals^n$ for which $x\tr M x \ge \evmax(M) - \eps$ with probability at least $1-\del$.
\end{lemma}

The proof relies on the analysis of the Lanczos method provided by \cite{kuczynski1992estimating}.

\begin{proof}

The statement of the lemma is proved in Theorem~4.2 of \cite{kuczynski1992estimating} when $M$ is a positive semidefinite matrix with $\norm{M}_2 \le 1$.
To prove the lemma for an arbitrary matrix $M$ with $\norm{M}_2 \le \lambda$, consider the matrix $M' = \frac{1}{2\lambda} M + \frac{1}{2} I$ which is positive semidefinite with $\norm{M'}_2 \le 1$.
If we apply the Lanczos method with error parameter~$\eps' = \frac{1}{2\lambda}\eps$, we obtain with high probability a unit vector $x$ such that $x^{\top}M'x \ge \evmax(M') - \eps'$.
Hence,
\begin{align*}
	\frac{1}{2\lambda} x\tr M x + \frac{1}{2}
	= x^{\top}M'x 
	\ge \evmax(M') - \eps' 
	= \frac{1}{2\lambda} \evmax(M) + \frac{1}{2} - \frac{1}{2\lambda}\eps ~,
\end{align*}
so that $x\tr M x \ge \evmax(M) - \eps$, as required.
\end{proof}

\bibliographystyle{abbrvnat}
\bibliography{bib}

\end{document}